\documentclass[a4paper]{article}

\usepackage{microtype} % if unwanted, comment out
\usepackage[utf8]{inputenc}

\usepackage[margin=1.3in]{geometry}

\usepackage{authblk}
\usepackage{amsthm,amssymb,amsfonts,amsmath}
\usepackage{amssymb,amsfonts,amsmath}
\usepackage{xcolor}
\usepackage{graphicx}
\graphicspath{{figures/}}

\usepackage{thm-restate}
\usepackage{url}
\usepackage{faktor}
\usepackage{tikz-cd}
\newtheorem{theorem}{Theorem}
\newtheorem{lemma}[theorem]{Lemma}

\newtheorem{definition}{Definition}

\newcommand{\R}{\mathbb{R}}
\newcommand{\N}{\mathbb{N}}

\title{Computing Enclosing Depth\footnote{This work is based on the Master thesis of the second author.}}
\date{}

\begin{document}
\author[1]{Bernd G\"{a}rtner}
\author[2]{Fatime Rasiti}
\author[3]{Patrick Schnider}
\affil[1]{Department of Computer Science, ETH Z\"{u}rich\\
  \texttt{gaertner@inf.ethz.ch}}
\affil[2]{Department of Mathematics, ETH Z\"{u}rich\\
  \texttt{frasiti@student.ethz.ch}}
\affil[3]{Department of Computer Science, ETH Z\"{u}rich\\
  \texttt{patrick.schnider@inf.ethz.ch}}

\maketitle

\begin{abstract}
Enclosing depth is a recently introduced depth measure which gives a lower bound to many depth measures studied in the literature. So far, enclosing depth has only been studied from a combinatorial perspective. In this work, we give the first algorithms to compute the enclosing depth of a query point with respect to a data point set in any dimension. In the plane we are able to optimize the algorithm to get a runtime of $O(n\log n)$. In constant dimension, our algorithms still run in polynomial time.
\end{abstract}

\section{Introduction}

Medians play an important role in statistics. In contrast to the mean value of some given data, the median depends only on the order of the data points and not on their exact positions. Hence, it is robust against outliers. As data sets are multidimensional in many cases, we are interested in an extension of the term 'median' to higher dimensions. Since there is no clear order of the data points, there are various generalizations of the median to higher dimensions \cite{aloupis, liu, moser}. In order to define the median of some data, the notion of depth of a query point has been introduced. A median is then a query point with the highest depth. Many depth measures only depend on the relative positions of the data points, just like the median, making them again robust against outliers.

After the first depth measure was introduced by Tukey \cite{tukey} (and is therefore known as Tukey depth), Donoho and Gasko \cite{donoho} established the idea of a multidimensional median as a deepest point relative to the data points. Various depth measures with different properties have since been introduced, such as simplicial depth \cite{regina} and convex hull peeling depth \cite{aloupis}.

Depth measures are an important tool in Computer Science for example in geometric matching, pattern matching, clustering \cite{andreas, rebecka, ida} and shape fitting applications \cite{aiger}. Since depth measures give a way to compute medians of data points they also find applications in Statistics such as data visualization \cite{tukey} and regression analysis \cite{hubert, hyperplane_depth}.

\begin{definition}[Depth measure]
Let $d \in \N$ and $(\R^d)^S$ be the family of all finite point sets in $\R^d$. A depth measure is a function $D: (\R^d)^S\times \R^d \rightarrow \R_{\geq 0}$, $(S,q) \mapsto D(S,q)$. In particular, the function $D$ assigns to a given finite point set $S$ and a query point $q$ a value, which describes how deep the query point $q$ lies within the data set $S$.
\end{definition}

Assume we are given a data set $S$. Consider all hyperplanes spanned by the points of S. This arrangement A of hyperplanes divides $\R^d$ into connected components of $\R^d\backslash A$. We call these connected components cells. A depth measure where all points in a cell have the same depth is called \emph{combinatorial}.

Aloupis et al. \cite{aloupis} used the fact that simplicial depth is a combinatorial depth measure to compute the simplicial median for $d=2$ in $O(n^4)$ time. Sachini \cite{sachini} modified this algorithm to compute the simplicial depth for the whole plane in $O(n^4)$ time. For the case $d=3$ there are various algorithms that compute the simplicial depth of a single query point in $O(n^2)$ \cite{cheng}. Cheng and Ouyang \cite{cheng} discussed an extension of this algorithm to compute simplicial depth in $\R^4$ in $O(n^4)$ time. Afshani et al. \cite{afshani} later introduced methods to compute simplicial depth in $O(n^d \log n)$ time for $d>4$. This bound was improved by Pilz et al. \cite{pilz} to $O(n^{d-1})$.

Another well studied combinatorial depth measure is Tukey depth, also known as halfspace depth.
For the case $d=2$, Aloupis et al. \cite{aloupis2002lower} gave a worst case lower bound of $\Omega(n \log n)$ for computing the Tukey depth of an arbitrary query point $q$ with respect to a given point set $S$ of size $n$. In fact, the Tukey depth of a query point relative to a point set of size $n$ can be computed in $O(n\log n)$ time \cite{ruts}. There are different approaches to compute the Tukey depth in different dimensions \cite{chen, bremer, peter}. The algorithm of Rousseeuw et al. \cite{peter} to compute the Tukey depth of a query point in $\R^d$ for $d>2$ has a run time of $O(n^{d-1} \log n)$.

Studying more general families of combinatorial depth measures, Schnider introduced the notion of \emph{enclosing depth}, which turns out to be a natural lower bounds for many combinatorial depth measures \cite{patrick}. In this work, we will focus on enclosing depth and provide algorithms to compute it.% For the formal definition of this depth measure we first need to define enclosingness.

\begin{definition}[k-enclosing]
Let $S$ be a point set of size $(d+1)k$ in $\mathbb{R}^d$ and $q$ a query point. If $S$ can be partitioned into $d+1$ pairwise disjoint subsets $S_1,...,S_{d+1}$, each of size $k$, such that for any transversal $p_1\in S_1,...,p_{d+1}\in S_{d+1}$ the point $q$ lies in the convex hull of $p_1,...,p_{d+1}$, then we say that $S_1,...,S_{d+1}$  $k$-encloses the point $q$. (Figure \ref{fig:transversals})
\begin{figure}[ht]
    \centering
    \includegraphics[width=0.6\textwidth]{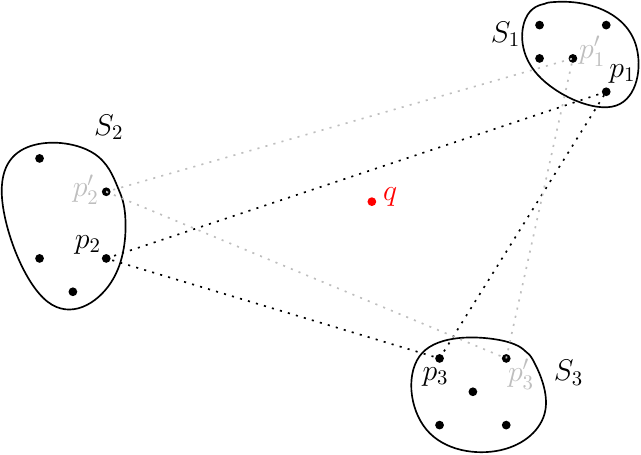}
    \caption{\label{fig:transversals} The set $S = S_1 \cup S_2 \cup S_3$ 5-encloses the query point $q$ in $\mathbb{R}^2$.}
    \end{figure}
\end{definition}

\begin{definition}[Enclosing Depth]
Let $S$ be a finite point set in $\mathbb{R}^d$ and $q$ be a query point. The enclosing depth of $q$ with respect to $S$ is the maximum $k$ such that there exist subsets $S_1,...,S_{d+1}$ of $S$ which $k$-enclose $q$. We denote it by $ED(S,q)$. (Figure \ref{fig:enclosingdepth})
\end{definition}

%\begin{example}
%In $\mathbb{R}^2$, i.e., $d=2$, we need to find three disjoint subsets of the same size that enclose the query point $q$. In Figure \ref{fig:enclosingdepth} the query point $q$ has an enclosing depth of at least 5, since there are three disjoint subsets of the given point set $S$, each of size 5, that enclose the point $q$.
\begin{figure}[ht]
    \centering
    \includegraphics[width=0.6\textwidth]{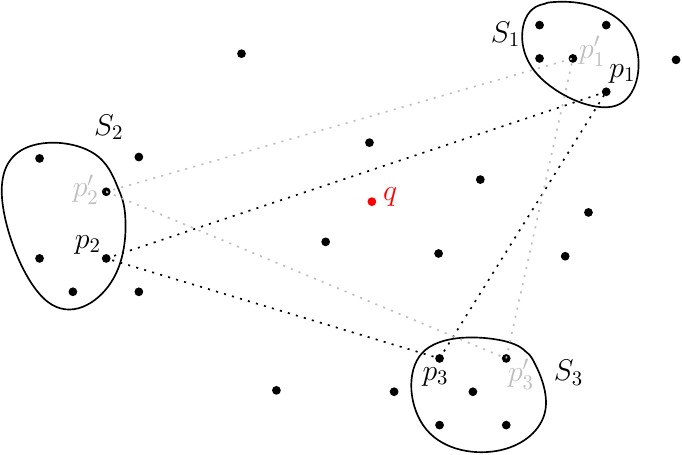}
    \caption{\label{fig:enclosingdepth} The enclosing depth of the query point $q$ is at least 5.}
    \end{figure}
%\end{example}

In this work, we present algorithms to compute the enclosing depth of a query point $q$ with respect to a data set $S$ in $\R^2$ (Section \ref{sec:plane}) as well as in general dimension (Section \ref{sec:higherd}).

\section{The planar case}\label{sec:plane}

Before describing our algorithm, we introduce some combinatorial lemmas that will be helpful in proving the correctness of our algorithm. For these lemmas, we will assume that the query point $q$ is the origin and that the data point set $S$ lies on the unit circle. Combinatorially, this is not a restriction, as the following lemma shows.

\begin{lemma}\label{lem:circle}
Let $S=\{s_1,\ldots,s_n\}\subset\R^2$ be a data point set and $q\in\R^2$ a query point such that $S\cup\{q\}$ is in general position. Denote by $S'=\{s'_1,\ldots,s'_n\}$ the point set defined by centrally projecting each point in $S$ to a circle of unit radius with center $q$. Then $q$ is in the convex hull of $s_i,s_j,s_k$ if and only if it is in the convex hull of $s'_i,s'_j,s'_k$.
\end{lemma}

\begin{proof}
Assume that $q$ is not in the convex hull of $a,b,c$. Then there is a line $\ell$ through $q$ having all of $a,b,c$ on the same side. This is invariant under central projection from $q$.
\end{proof}

Let now $S=\{s_1,\ldots,s_n\}$ on the unit circle be ordered in counter-clockwise direction. By an \emph{interval} $[s_a,s_b]$ we denote all the points in $S$ that lie between $s_a$ and $s_b$, that is,

\[
[s_a,s_b] = \left\{
\begin{array}{ll}
\{s\in S\mid s_a\leq s\leq s_b\} & s_a\leq s_b \\
\{s\in S\mid s\geq s_a \text{ or } s\leq s_b\} & \, s_a>s_b \\
\end{array}
\right.
\]

In the following, we write indices modulo $n$, that is, $s_i=s_{i-n}$ for $i\geq n$.
We show that in order to find $k$-enclosing sets we can restrict our attention to intervals and their endpoints.

\begin{lemma}\label{lem:intervals}
Let $a_1,a_2,b_1,b_2,c_1,c_2\in S$ such that the intervals $[a_1,a_2]$, $[b_1,b_2]$ and $[c_1,c_2]$ are pairwise disjoint and for every choice of $a\in\{a_1,a_2\}$, $b\in\{b_1,b_2\}$, $c\in\{c_1,c_2\}$ the origin lies in the convex hull of $a,b,c$. Then for every choice $a\in [a_1,a_2]$, $b\in [b_1,b_2]$, $c\in [c_1,c_2]$ the origin lies in the convex hull of $a,b,c$.
\end{lemma}

\begin{proof}
Assume for the sake of contradiction that the convex hull of $a,b,c$ does not contain the origin and let $\ell$ be a line through the origin that has all of $a,b,c$ on the same side $\ell^+$. As the intervals are pairwise disjoint, one of $a_1$ or $a_2$ must also be in $\ell^+$. Thus the convex hull of $a_i,b,c$ does not contain the origin for some $i\in\{1,2\}$. Repeating this argument two more times, we find that one of the 8 triangles spanned by the endpoints does not contain the origin, which is a contradiction to the assumptions of the lemma.
\end{proof}

We want to restrict the considered intervals even further. To this end, we define for each point $s\in S$ its \emph{opposite neighbors} $s^{(r)}$ and $s^{(\ell)}$ as the last point in $S$ before $-s$ and the first point in $S$ after $-s$, respectively, see Figure \ref{fig:opposite_nbs} for an illustration.

\begin{figure}[ht]
    \centering
    \includegraphics[width=0.6\textwidth]{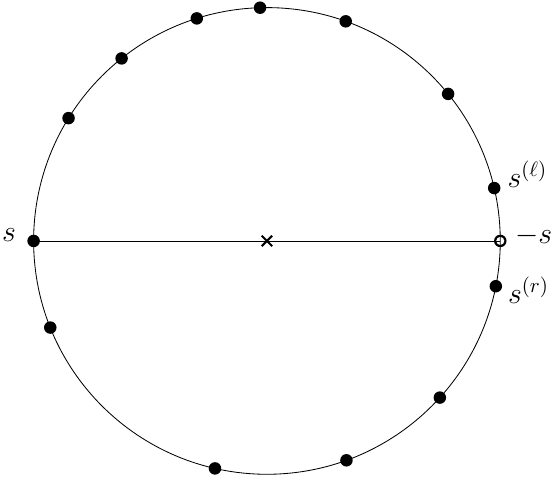}
    \caption{The opposite neighbors $s^{(r)}$ and $s^{(\ell)}$ of a point $s\in S$.}
    \label{fig:opposite_nbs}
\end{figure}

\begin{lemma}\label{lem:nbs}
Let $S_1=[s_i,s_{i+k}]$, $S_2$ and $S_3$ be subsets of $S$ that $(k+1)$-enclose the origin. Denote $s_j=s_i^{(r)}$ and $s_{m}=s_{i+k}^{(\ell)}$. Then $S_1$, $S'_2:=[s_{j-k},s_j]$ and $S'_3:=[s_{m},s_{m+k}]$ also $(k+1)$-enclose the origin.
\end{lemma}

\begin{proof}
See Figure \ref{fig:nbs_proof} for an illustration of the proof.
We first note that (up to relabeling) we must have that $S_2\subset [s_{i+k+1},s_j]$ and $S_3\subset [s_m,s_{i-1}]$. Indeed, both lines through $s_i$ and the origin, as well as through $s_{i+k}$ and the origin must separate $S_2$ and $S_3$, as otherwise there would be a choice $a\in S_1$, $b\in S_2$, $c\in S_3$ whose convex hull does not contain the origin. Thus, we can write $S_2=[s_b,s_{b+k}]$ for $i+k+1\leq b\leq j-k$ and similarly $S_2=[s_{c-k},s_{c}]$ for $m+k\leq c\leq i-1$. In particular both $s_i,s_b,s_c$ and $s_{i+k},s_b,s_c$ contain the origin.

We claim that $s_i,s_b,s_m$ also contains the origin. Assume it does not. Then either the line $\ell_i$ through $s_i$ and the origin or the line $\ell_b$ through $s_b$ and the origin must separate $s_c$ and $s_m$. By construction, $\ell_i$ has both $s_c$ and $s_m$ on the same side, so it must be $\ell_b$. But this would imply that $m>c$, which is a contradiction. A symmetric argument also shows that $s_{i+1}, s_b,s_m$ contains the origin. Analogously it can be shown that $s_i,s_j,s_c$ as well as $s_{i+1},s_j,s_c$ contains the origin. Finally, by construction $s_i,s_j,s_m$ as well as $s_{i+k},s_j,s_m$ contains the origin. The statement now follows from Lemma \ref{lem:intervals} and the fact that $b\leq j-k$ and $m+k\leq c$.
\end{proof}

\begin{figure}[ht]
    \centering
    \includegraphics[width=0.6\textwidth]{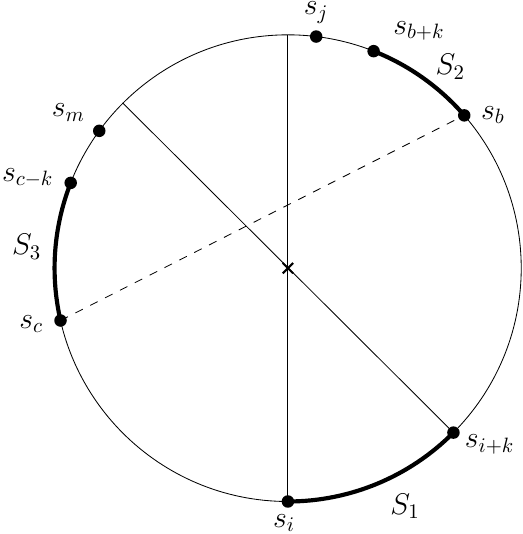}
    \caption{An illustration of the proof of Lemma \ref{fig:nbs_proof}.}
    \label{fig:nbs_proof}
\end{figure}

\textbf{Description of the algorithm:}
We are given a set $S$ of $n$ data points in the plane and a query point $q$. In a pre-processing step we first sort the points radially around $q$, giving a counter-clockwise order $s_1,\ldots,s_n$ on $S$ and then for each $s\in S$ we compute $s^{(r)}$ and $s^{(\ell)}$ using binary search.
For the main part of the algorithm we run the following subroutine, which for a given integer $k$ checks whether there are three sets that $(k+1)$-enclose $q$: for each $s_i\in S$, with $s_j=s_i^{(r)}$ and $s_{m}=s_{i+k}^{(\ell)}$ the subroutine checks whether all 8 triangles $a,b,c$ for $a\in [s_i,s_{i+k}]$, $b\in [s_{j-k},s_j]$, $c\in [s_{m},s_{m+k}]$ contain $q$ and the intervals are pairwise disjoint, returning TRUE if this holds for some $s_i$, and FALSE otherwise. By doing a binary search over the values of $k\in\{0,\ldots,n\}$ we find the largest value $k$ for which the subroutine returns true and return $(k+1)$.

\begin{theorem}
The above algorithm computes the enclosing depth of $q$ with respect to $S\subset\R^2$ in time $O(n\log n)$.
\end{theorem}

\begin{proof}
We first show the correctness of the algorithm. It follows from Lemma \ref{lem:intervals} that if the subroutine returns TRUE then the considered intervals are indeed $(k+1)$-enclosing. On the other hand, if there are $(k+1)$-enclosing sets, by Lemma \ref{lem:nbs} the subroutine will find them.

As for the runtime, we can sort in time $O(n\log n)$. After this, we perform $2n$ binary searches, each taking $O(\log n)$ time, thus the total runtime of the pre-processing step is $O(n\log n)$. For the runtime of the subroutine, we notice that for each choice of $s_i$ the required checks can be done in time $O(1)$, so the runtime of the subroutine is $O(n)$. As we call it for $O(\log n)$ many values of $k$ we get the desired runtime.
\end{proof}

\section{Higher dimensions}\label{sec:higherd}

Our algorithm in general dimension is based on the following observation:

\begin{lemma}[Lemma 20 in \cite{patrick}]
\label{lem:halfspaces}
Let $S_1,\ldots, S_{d+1}\subset S\subset\R^d$ be point sets which enclose a point $q$, where $S\cup\{q\}$ is in general position.
Then there are $d+1$ closed halfspaces $H^{-}_1,\ldots,H^{-}_{d+1}$ such that each $H^{-}_i$ contains $q$ on its boundary, $H^{-}_i\cap S=S_i$ for each $i$ and $H^{-}_1\cup\ldots \cup H^{-}_{d+1}=\mathbb{R}^d$.
\end{lemma}

Denoting by $H^+_i$ the complement of $H^{-}_i$ we show below that we get $S_i\subset\bigcap_{j\neq i}H^+_j$ and $\bigcap_i H^+_i=\{q\}$.
As we also show later, given enclosing sets $S_1,\ldots, S_{d+1}$ we can rotate the halfspaces $H^+_i$ to get closed halfspaces $H'_i$ whose boundaries contain $q$ and $d-1$ points of $S$ and for which $H'_i\cap S=H^+_i\cap S$ and $\bigcap_i H'_i=\{q\}$. On the other hand, given halfspaces $H'_1,\ldots,H'_{d+1}$, each boundary containing $q$ and $d-1$ points of $S$ with $\bigcap_i H'_i=\{q\}$, defining $S_i:=\bigcap_{j\neq i}(H'_j\cap S)$ we will show that for every transversal $p_1\in S_1,...,p_{d+1}\in S_{d+1}$ the point $q$ lies in the convex hull of $p_1,...,p_{d+1}$. Combining these facts, we get the following strengthening of Lemma \ref{lem:halfspaces}.

\begin{lemma}
\label{lem:halfspaces_strong}
Let $S_1,\ldots, S_{d+1}\subset S\subset\R^{d}$ and $q\in\R^{d}$ such that $S\cup\{q\}$ is in general position.
Then $S_1,\ldots,S_{d+1}$ enclose $q$ if and only if there are $d+1$ halfspaces $H'_1,\ldots H'_{d+1}$ whose boundaries contain $q$ and $d-1$ points of $S$, and for which $S_i\subset\bigcap_{j\neq i}H^+_j$ and $\bigcap_i H^+_i=\{q\}$.
\end{lemma}

Let us now formally prove all the steps outlined above.
We say that $d+1$ halfspaces in $\R^d$ are in \emph{general position} if for any $k\leq d$ of them their boundary hyperplanes intersect in a $(d-k)$-dimensional affine subspace. This is equivalent to saying that the normal vectors of the bounding hyperplanes are affinely independent. The halfspaces $H^{-}_1,\ldots,H^{-}_{d+1}$ given by Lemma \ref{lem:halfspaces} are in general position, as an investigation of the proof in \cite{patrick} shows. In fact, this also follows from Lemma \ref{lem:trivial_intersection}. Our first lemma relates the union of halfspaces to their intersection.

The following is well-known, but we state it as it will be helpful for us later.

\begin{lemma}\label{lem:union_vs_intersection}
Let $H_1,\ldots,H_{d+1}$ be $d+1$ closed halfspaces in $\R^d$ in general position, all of which contain a point $q$ on their boundary. Then $\bigcup_i H_i=\R^d$ if and only if $\bigcap_i H_i=\{q\}$.
\end{lemma}

The next lemma relates the hyperplanes $H^{-}_i$ of Lemma \ref{lem:halfspaces} to their complements $H^+_i$.

\begin{lemma}\label{lem:reformulation}
Let $S=S_1\cup\ldots\cup S_{d+1}$ and let $H^{-}_1,\ldots,H^{-}_{d+1}$ be closed halfspaces in general position whose boundaries all contain $q$. Denote by $H^+_1,\ldots H^+_{d+1}$ the closures of their complements.
We have that $H^{-}_i\cap S=S_i$ for every $i\in\{1,\ldots,d+1\}$ and $\bigcup_i H^{-}_i=\R^d$ if and only if $\bigcap_{j\neq i}H^+_j\cap S=S_i$ for every $i\in\{1,\ldots,d+1\}$ and $\bigcap_i H^+_i=\{q\}$.
\end{lemma}

\begin{proof}
We first show that $\bigcap_i H^+_i=\{q\}$ if and only if $\bigcap_i H^{-}_i=\{q\}$. Assume that $\bigcap_i H^+_i=\{q\}$. By definition we have $q\in\bigcap_i H^{-}_i$. It remains to show that there is no other point in the intersection. For this, we note that the intersection of $d+1$ closed halfspaces in general position in $\R^d$ is a convex set and as all boundary hyperplanes contain $q$ it is either only $q$ or a cone with apex $q$. However, if it is a cone then any point in the interior of the cone must lie in the interior of $H^{-}_i$ for every $i$, so it does not lie in $H^+_i$ for any $i$, which means that $\bigcup_i H^+_i\neq\R^d$. By Lemma \ref{lem:union_vs_intersection}, this is a contradiction to $\bigcap_i H^+_i=\{q\}$. The other direction is analogous. In particular, by Lemma \ref{lem:union_vs_intersection} we get that $\bigcup_i H^{-}_i=\R^d$ if and only if $\bigcap_i H^+_i=\{q\}$.

It remains to show that $H^{-}_i\cap S=S_i$ for every $i\in\{1,\ldots,d+1\}$ if and only if $\bigcap_{j\neq i}H^+_j\cap S=S_i$. This follows immediately from the fact that $H^+_i$ is the complement of $H^{-}_i$.
\end{proof}

The next lemma shows that the condition $\bigcap_i H^+_i=\{q\}$ is an immediate consequence of enclosingness.

\begin{lemma}\label{lem:trivial_intersection}
Let $s_1,\ldots,s_{d+1},q$ be $d+2$ points in $\R^d$ in general position such that $q$ lies in the convex hull of $S=\{s_1,\ldots,s_{d+1}\}$. Let $H^+_1,\ldots,H^+_{d+1}$ be closed halfspaces whose boundaries contain $q$ and such that $\bigcap_{j\neq i}H^+_j\cap S=s_i$ for every $i\in\{1,\ldots,d+1\}$. Then $\bigcap_i H^+_i=\{q\}$.
\end{lemma}

\begin{proof}
Consider the simplex spanned by $s_1,\ldots,s_{d+1}$ and denote by $F_i$ the facet spanned by all points except $s_i$. Note that $H^+_i$ contains $F_i$ and that the $H^+_i$ must be in general position. This implies that $\bigcup_i H^+_i=\R^d$. The statement not follows from Lemma \ref{lem:union_vs_intersection}.
\end{proof}

We are now ready to show that we can rotate the $H^+_i$ to not only contain $q$ on the boundary but also $(d-1)$ points of $S$.

\begin{lemma}\label{lem:rotation}
Let $S_1,\ldots, S_{d+1}\subset S\subset\R^d$ be point sets which enclose a point $q$, where $S\cup\{q\}$ is in general position.
Then there are $d+1$ closed halfspaces $H^{+}_1,\ldots,H^+_{d+1}$ such that each $H^{+}_i$ contains $q$ and $d-1$ points of $S$ on its boundary, $S_i\subseteq\bigcap_{j\neq i} H^{+}_i\cap S$ for each $i$ and $\bigcap_i H^+_i=\{q\}$.
\end{lemma}

\begin{proof}
We start with the halfspaces $H^+_1,\ldots,H^+_{d-1}$, which are the complements of the halfspaces given by Lemma \ref{lem:halfspaces}.
Assume the boundary of $H^+_i$ contains $q$ but fewer than $d-1$ points of $S$. In particular, it contains fewer than $d$ points, so we can choose a $(d-2)$-dimensional axis of rotation through $q$ and all points of $S$ it contains. We rotate $H^+_i$ around this axis until it hits some other point, resulting in a rotated halfspace with one more point on the boundary. Also note that no point has been removed from $H^+_i$ during the rotation, so we maintain $S_i\subseteq\bigcap_{j\neq i} H^{+}_i\cap S$ for each $i$. Further it follows from Lemma \ref{lem:trivial_intersection} that $\bigcap_i H^+_i=\{q\}$. The lemma now follows by repeating this process until all boundaries contain the required number of points.
\end{proof}

Finally, we show the other direction of Lemma \ref{lem:halfspaces_strong}.

\begin{lemma}\label{lem:simplex_enclosing}
Let $S$ be a set of points in $\R^d$ and $q$ a query point, such that $S\cup\{q\}$ is in general position. Let $H^{+}_1,\ldots,H^+_{d+1}$ be closed halfspaces such that each $H^{+}_i$ contains $q$ and $d-1$ points of $S$ on its boundary and $\bigcap_i H^+_i=\{q\}$. Let $S_i=\bigcap_{j\neq i} H^{+}_i\cap S$. Then $S_1,\ldots, S_{d+1}$ enclose $q$.
\end{lemma}

\begin{proof}
For each $H^+_i$ consider the $(d-2)$-dimensional rotation axis spanned by the $d-1$ points on its boundary and rotate a small amount such that $q$ is not in $H^+_i$ anymore. The resulting halfspaces are still in general position but they now have an empty intersection and thus their complements intersect in a simplex $\Delta$ containing $q$. Let $v_i$ be the vertex of $\Delta$ that is the intersection of all bounding hyperplanes except the one of $H^+_i$. Then $\bigcap_{j\neq i} H^{+}_i$ is a cone $C_i$ with apex $v_i$ which contains $S_i$. As we have that for every choice $p_1\in C_1,\ldots,p_{d+1}\in C_{d+1}$ that the convex hull of $p_1,\ldots,p_{d+1}$ contains $\Delta$ it follows that Then $S_1,\ldots, S_{d+1}$ enclose $q$.
\end{proof}

\textbf{Description of the algorithm:}
For each choice of $d-1$ points in $S$, consider the two halfspaces defined by the hyperplane through $q$ and these $d-1$ points. This defines a set $\mathcal{H}$ of $2\binom{n}{d-1}$ halfspaces. For any $d+1$ halfspaces $H_1,\ldots,H_{d+1}\in\mathcal{H}$ first check if $\bigcap_i H_i=\{q\}$. If not, continue with the next choice, otherwise count for each $i$ the number of points of $S$ in $\bigcap_{j\neq i} H_j$ and denote by $k$ the smallest of the $d+1$ numbers. In the end, return the largest such $k$ encountered over all choices of $d+1$ halfspaces in $\mathcal{H}$.

\begin{theorem}
The above algorithm computes the enclosing depth of $q$ with respect to $S\subset\R^d$ in time $O(n^{d^2})$.
\end{theorem}

\begin{proof}
The correctness follows from Lemma \ref{lem:halfspaces_strong}. As for the runtime, we have $|\mathcal{H}|\in O(n^{d-1})$ out of which we choose sets of $d+1$ elements, so there are $O(n^{(d-1)(d+1})$ sets of halfspaces that we consider. For each set the check and the counting can be done in time $O(n)$ so the total runtime is $O(n^{(d-1)(d+1)+1})=O(n^{d^2})$. 
\end{proof}

\section{Conclusion}
We have given two algorithms to compute the enclosing depth of a query point $q$ with respect to a data point set $S$, one for the plane and one in general dimension. The planar algorithm matches the runtimes of computing many other depth measures. For some measures, such as Tukey depth, matching lower bounds for the computation time have been shown \cite{aloupis2002lower}. It would be interesting to adapt these lower bounds to enclosing depth.

In higher dimension, many depth measures can be computed in time $O(n^{d-1})$, which is significantly faster than the runtime of our algorithm. We believe that enclosing depth can be computed more efficiently as well, but some additional ideas are likely required for this.

Finally, there are other natural algorithmic problems for depth measures even in the plane, such as computing the depth of the entire plane or finding a deepest query point. Using our algorithm and the fact that the arrangement spanned by the lines through all pairs of data points has $O(n^4)$ cells we can solve both those problems in $O(n^5\log n)$, but we again believe that this is not optimal.

%\subparagraph*{Acknowledgments.} We thank the organizers for the tasty cookies.

\bibliography{refs}

\end{document}